\documentclass[draft]{amsart}

\usepackage{amssymb}

\usepackage{enumerate}

\usepackage{amsmath}
\usepackage{amsthm}
\usepackage{tikz} 
\usepackage{float}
\usepackage{mathtools}
\usepackage{graphicx}

\usepackage{pgf,tikz}\usepackage{mathrsfs}\usetikzlibrary{arrows}

\usetikzlibrary{matrix,fit}

\numberwithin{equation}{section}

\newtheorem{thm}{Theorem}[section]

\newtheorem{cor}[thm]{Corollary}
\newtheorem{lem}[thm]{Lemma}
\newtheorem{rem}[thm]{Remark}

\theoremstyle{definition}
\newtheorem{defn}{Definition}[section]
\newtheorem{ex}[defn]{Example}

\begin{document}
\definecolor{ffqqqq}{rgb}{1.,0.,0.}
\definecolor{qqqqff}{rgb}{0.,0.,1.}
\definecolor{ududff}{rgb}{0.30196078431372547,0.30196078431372547,1.}
\definecolor{ffqqtt}{rgb}{1.,0.,0.2}
\definecolor{zzttqq}{rgb}{0.6,0.2,0.}
\definecolor{uuuuuu}{rgb}{0.26666666666666666,0.26666666666666666,0.26666666666666666}

\author{Sa\'ul A. Blanco and Charles Buehrle}
\title[Relations on prefix reversal generators]{Some relations on prefix reversal generators of the symmetric and hyperoctahedral group}
\address{Department of Computer Science, Indiana University, Bloomington, IN 47408}
\email{sblancor@indiana.edu}
\address{ Department of Mathematics, Physics, and Computer Studies, Notre Dame of Maryland University, Baltimore, MD 21210}
\email{cbuehrle@ndm.edu}

\bibliographystyle{plain}

\date{June 7, 2018}

\begin{abstract}
The pancake problem is concerned with sorting a permutation (a stack of pancakes of different diameter) using only prefix reversals (spatula flips). Although the problem description belies simplicity, an exact formula for the maximum number of flips needed to sort $n$ pancakes has been elusive. 

In this paper we present a different approach to the pancake problem, as a word problem on the symmetric group and hyperoctahedral group. Pancake flips are considered as generators and we study the relations satisfied by them. We completely describe the order of the product of any two of these generators, and provide some partial results on the order of the product of any three generators. Connections to the pancake graph of the hyperoctahedral group are also drawn.
\end{abstract}

\maketitle

\section{Introduction}\label{s:intro}
The first appearance of the pancake problem in print was in the Problems and Solutions section of the December 1975 \emph{Monthly}~\cite{Dweighter75}. 
	\begin{quote}
		The chef in our place is sloppy, and when he prepares a stack of pancakes they come out all different sizes. Therefore, when I deliver them to a customer, on the way to the table I rearrange them (so that the smallest winds up on top, and so on, down to the largest on the bottom) by grabbing several from the top and flipping them over, repeating this (varying the number I flip) as many times as necessary. If there are $n$ pancakes, what is the maximum number of flips (as a function $f(n)$ of $n$) that I will ever have to use to rearrange them?
	\end{quote}

The problem of determining the maximum number of flips that are ever needed to sort a stack of $n$ pancakes is known as the \emph{pancake problem}, and the $f(n)$ is known as the \emph{pancake number}.
	
This initial posing of the problem was made by Jacob E. Goodman of the City College of New York, under the pseudonym Harry Dweighter (a pun for ``harried waiter"). In~\cite{DGJL77}, as a commentary to the problem formulation in~\cite{Dweighter75}, Michael R. Garey, David S. Johnson, and Shen Lin from Bell Labs gave the first upper and lower bound to the the pancake number: 
\[n+1 \leq f(n) \leq 2n-6 \text{ for } n \geq 7.\]

Subsequent results have been successful in tightening these bounds. The first significant tightening of the bounds was described in the work of William H. Gates and Christos H. Papadimitriou~\cite{GatesPapa}, which incidentally is the only academic paper Gates ever wrote. The best upper and lower bound known today for the general case appeared in ~\cite{Chitt} and~\cite{HeySud}, respectively. Combined, one has that
\[
15\left\lfloor \frac{n}{14}\right\rfloor \leq f (n) \leq \frac{18n}{11}+O(1).
\]

Computing the pancake number for a given $n$ is a complicated task. To our knowledge, the exact value of $f(n)$ is only known for $1\leq n\leq 19$ (see~\cite{Asai2006, Cibulka, CohenBlum, HeySud, KKS}). In fact, determining the minimum number needed to sort a stack of pancakes is an NP-hard problem~\cite{BulFerRusu}, though 2-approximation algorithms exists~\cite{Fischer2005}.

The pancake problem has connections to parallel computing, in particular in the design of symmetric interconnection networks (networks used to route data between the processors in a multiprocessor computing system) where the so-called \emph{pancake graph}, the Cayley graph of the symmetric group under prefix reversals, gives a model for processor interconnections (see~\cite{AkersKrish, Qiu1991}). A pancake network is shown in Figure~\ref{f:network}. One can also define a \emph{burnt pancake graph} on signed permutations (See Section~\ref{s:notation} for the necessary definitions), and we exhibit one in Figure~\ref{f:networkb}.

\begin{figure}
\begin{center}
\begin{tikzpicture}[scale=1.5,every node/.style={scale=1}]
		\node (e) at (-1.5,3.5) {$1234$};
			
		\node (2) at (-2.5,2.5) {$3214$};
		\node (1) at (-0.5,2.5) {$2134$};
		\node (3) at (1.5,3.5) {$4321$};
			
		\node (32) at (2.5,-2.5) {$4123$};
		\node (12) at (-2.5,1.5) {$2314$};
		\node (21) at (-0.5,1.5) {$3124$};
		\node (31) at (-0.5,-2.5) {$4312$};
		\node (13) at (0.5,2.5) {$3421$};
		\node (23) at (2.5,2.5) {$2341$};

		\node (232) at (1.5,-3.5) {$2143$};						
		\node (132) at (2.5,-1.5) {$1423$};
		\node (312) at (-2.5,-1.5) {$4132$};
		\node (212) at (-1.5,0.5) {$1324$};
		\node (321) at (0.5,-1.5) {$4213$};
		\node (231) at (-0.5,-1.5) {$1342$};
		\node (131) at (-1.5,-3.5) {$3412$};
		\node (213) at (0.5,1.5) {$2431$};
		\node (313) at (0.5,-2.5) {$1243$};
		\node (123) at (2.5,1.5) {$3241$};
		\node (323) at (-2.5,-2.5) {$1432$};
			
		\node (1321) at (1.5,-0.5) {$2413$};
		\node (1231) at (-1.5,-0.5) {$3142$};
		\node (1213) at (1.5,0.5) {$4231$};
			
		\draw[line width=2.pt,red](e)--(1) (2)--(12) (3)--(13) (31)--(131) (21)--(212) (32)--(132) (23)--(123) (231)--(1231) (321)--(1321) (213)--(1213);
		\draw[line width=2.pt,blue] (e)--(2) (1)--(21) (3)--(23) (31)--(231) (12)--(212) (32)--(232) (13)--(213) (312)--(1231) (132)--(1321) (123)--(1213);
		\draw[line width=2.pt,purple!50!black] (e)--(3) (2)--(32) (21)--(321) (12)--(312) (23)--(323) (212)--(1213);
		\draw[line width=2.pt,purple!50!black,domain=-72:72] plot ({0.25+2*cos(\x)},{2.5*sin(\x)}) plot ({-0.25+2*cos(\x+180)},{2.5*sin(\x+180)});
		
		\draw[line width=2.pt] (232) edge[red] (313) 
		          (232) edge[purple!50!black] (131) 
		          (132) edge[purple!50!black] (123) 
		          (321) edge[blue] (313) 
		          (231) edge[purple!50!black] (213)
		          (312) edge[red] (323)
		          (131) edge[blue] (323)
		          (1321) edge[purple!50!black] (1231);	
	\end{tikzpicture}	
\caption{Pancake graph of $S_4$. The different colors indicate the different pancake generators.}
\label{f:network}
\end{center}
\end{figure}
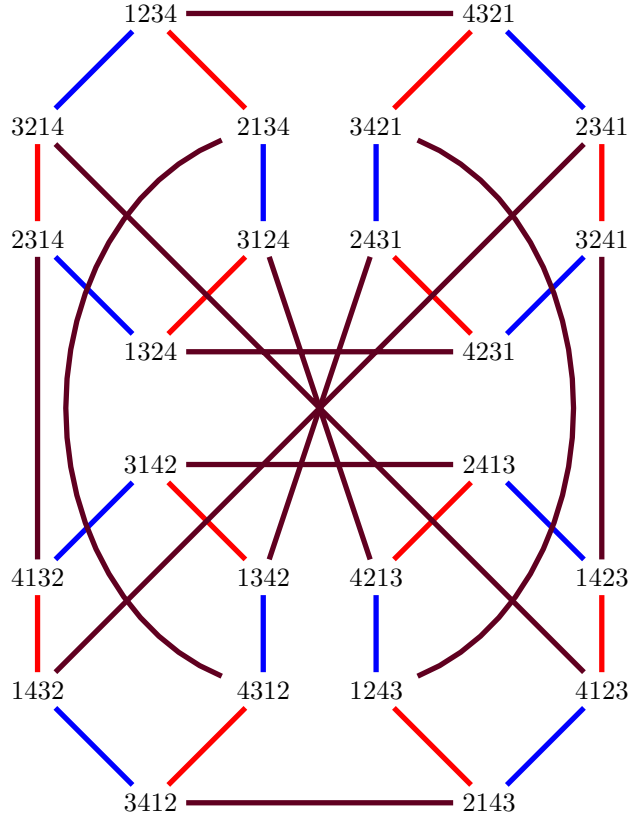

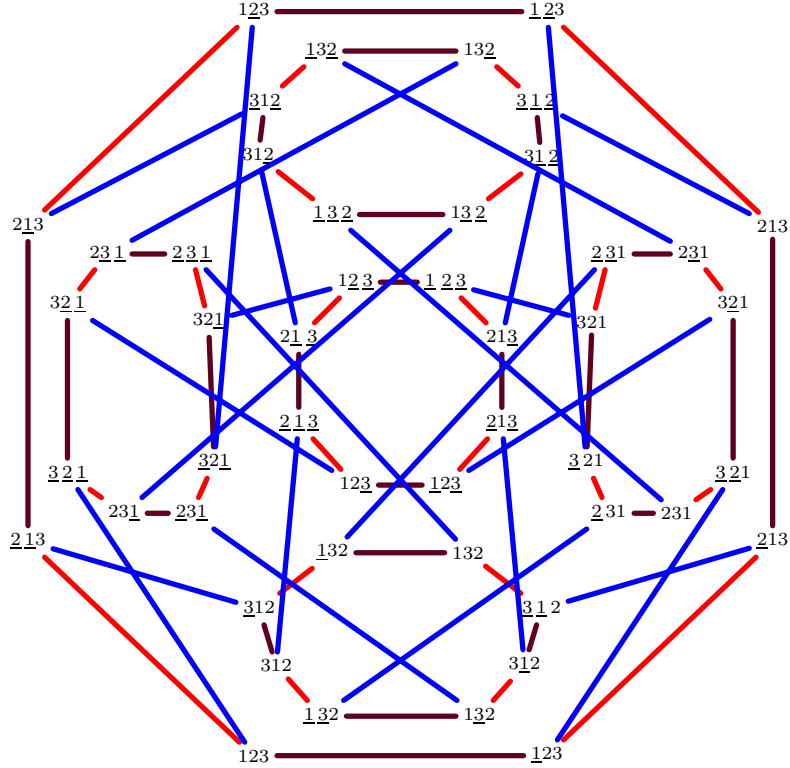
\begin{figure}
\begin{center}
	\begin{tikzpicture}[scale=0.15,line cap=round,line join=round,>=triangle 45,x=1.0cm,y=1.0cm]
	
	\begin{scriptsize}
	\node (e) at (-13,-33) {123};
	\node (0) at (13,-33) {\underline{1}23};
	\node (10) at (33,-14) {\underline{2}13};
	\node (010) at (33,14) {213};
	\node (1010) at (13,33) {$\underline{1}\,\underline{2}3$};
	\node (101) at (-13,33) {$1\underline{2}3$};
	\node (01) at (-33,14) {2\underline{1}3};
	\node (1) at  (-33,-14) {\underline{2}\,\underline{1}3};

	\node (202) at  (-4,-9) {12\underline{3}};
	\node (0202) at  (4,-9) {\underline{1}2\underline{3}};
	\node (10202) at  (9,-3.5) {\underline{2}1\underline{3}};
	\node (010202) at  (9,4) {21\underline{3}};
	\node (1010202) at  (4,9) {$\underline{1}\;\underline{2}\,\underline{3}$};
	\node (101202) at  (-4,9) {$1\underline{2}\,\underline{3}$};
	\node (01202) at  (-9,4) {$2\underline{1}\;\underline{3}$};
	\node (1202) at  (-9,-3.5) {$\underline{2}\,\underline{1}\,\underline{3}$};

	\node (21) at  (-12.5,-20) {\underline{3}12};
	\node (121) at  (-6,-15) {\underline{1}32};
	\node (0121) at  (6,-15) {132};
	\node (10121) at  (12.5,-20) {\underline{3}\,\underline{1}\,2};
	\node (010121) at  (11,-25) {3\underline{1}2};
	\node (01021) at  (7,-29.5) {1\underline{3}2};
	\node (1021) at  (-7,-29.5) {\underline{1}\,\underline{3}2};
	\node (021) at  (-11,-25) {312};
	
	\node (212) at (6,15) {1\underline{3}\,\underline{2}};
	\node (0212) at (-6,15) {\underline{1}\,\underline{3}\,\underline{2}};
	\node (10212) at (-12.6,20.2) {31\underline{2}};
	\node (010212) at (-12,25) {\underline{3}1\underline{2}};
	\node (1010212) at (-7,29.5) {\underline{1}3\underline{2}};
	\node (101212) at (7,29.5) {13\underline{2}};
	\node (01212) at (12,25) {\underline{3}\,\underline{1}\,\underline{2}};
	\node (1212) at (12.5,20) {3\underline{1}\,\underline{2}};
	
	\node (20) at  (29.5,-8) {\underline{3}\,\underline{2}1};
	\node (020) at  (29.5,7) {3\underline{2}1};
	\node (1020) at  (26,11.5) {2\underline{3}1};%
	\node (01020) at  (18.5,11.5) {\underline{2}\,\underline{3}1};
	\node (101020) at  (17,5.5) {321};
	\node (10120) at  (16.5,-7) {\underline{3}\,21};	
	\node (0120) at  (18.5,-11.5) {\underline{2}\,31};
	\node (120) at  (24.5,-11.5) {231};	

	\node (2) at  (-29.5,-8) {\underline{3}\,\underline{2}\,\underline{1}};
	\node (02) at  (-29.5,7) {3\underline{2}\,\underline{1}};
	\node (102) at  (-26,11.5) {2\underline{3}\,\underline{1}};%
	\node (0102) at  (-18.5,11.5) {\underline{2}\,\underline{3}\,\underline{1}};
	\node (10102) at  (-17,5.5) {32\underline{1}};
	\node (1012) at  (-16.5,-7) {\underline{3}2\underline{1}}; 
	\node (012) at  (-18.5,-11.5) {\underline{2}3\underline{1}};
	\node (12) at  (-24.5,-11.5) {23\underline{1}};	
	\end{scriptsize}
	
	\draw[line width=2,purple!50!black] (e)--(0) (10)--(010) (101)--(1010) (1)--(01) (202)--(0202) (10202)--(010202) 
	(101202)--(1010202) (1202)--(01202) (21)--(021) (121)--(0121) (10121)--(010121) (1021)--(01021) 
	(212)--(0212) (10212)--(010212) (101212)--(1010212) (1212)--(01212) (20)--(020) (1020)--(01020) 
	(10120)--(101020) (120)--(0120) (2)--(02) (102)--(0102) (1012)--(10102) (12)--(012);
	\draw[line width=2,red] (e)--(1) (0)--(10) (2)--(12) (01)--(101) (02)--(102) (20)--(120) 
	(21)--(121) (010)--(1010) (012)--(1012) (020)--(1020) (021)--(1021) (010212)--(1010212) 
	(0121)--(10121) (202)--(1202) (212)--(1212) (0202)--(10202) (010202)--(1010202) (01202)--(101202) 
	(0102)--(10102) (01021)--(010121) (01020)--(101020) (0120)--(10120) (01212)--(101212) (0212)--(10212);
	\draw[line width=2,blue] (e)--(2) (0)--(20) (1)--(21) (10)--(10121) (02)--(202) (010)--(01212) 
	(1010)--(10120) (101)--(1012) (01)--(010212) (020)--(0202)  (010121)--(10202) (10102)--(101202)
	(101020)--(1010202) (120)--(0212) (121)--(01020) (12)--(212) (0121)--(0102) (010202)--(1212)
	(01021)--(012) (1021)--(0120) (021)--(1202) (101212)--(102) (1010212)--(1020) (01202)--(10212);
	
\end{tikzpicture}
\caption{Pancake graph of $B_3$. Different edge colors indicate the different pancake generators.}
\label{f:networkb}
\end{center}
\end{figure}

The remainder of this paper is organized as follows. The basic terminology and notation are included in Section~\ref{s:notation}. In Section~\ref{s:snresults} we present the results corresponding to the pancake generators for $S_n$. A description of the Coxeter-like relations, that is, a description of the order of $f_if_j$ where $f_i,f_j$ are pancake generators. This description can be deduced from \cite[Lemma 1]{KM16}. Furthermore, we exhibit a partial description of the order of the product of three pancake generators and prove that the set of reflections generated by the pancake generators is the set of involutions in $S_n$. Finally, in Section~\ref{s:sbresults}, we exhibit a description of the order of the product of two pancake generators for signed permutations, which in turn gives the length of certain cycles in the burnt pancake graph.

\section{Terminology and Notation}\label{s:notation}

Let $S_{n}$ be the group of permutations of the set $[n]:=\{1,2,\ldots,n\}$ and denote by $e$ the identity permutation. The group $S_{n}$ is generated by the set $S:=\{s_{1},\ldots, s_{n-1}\}$ of adjacent transpositions; that is, $s_{i}=(i, i+1)$ in cycle notation. The set $S$ is subject to the relations $s_{i}^{2}=e$ for all $1\leq i\leq n-1$, $(s_{i}s_{i+1})^{3}=e$ for all $1\leq i\leq n-2$, and $s_{i}s_{j}=s_{j}s_{i}$ for all $i,j\in[n-1]$ with $|i-j|\geq 2$. It is well-known that the pair $(S_n,S)$ is a \emph{Coxeter system} (see~\cite{BjornerBrenti}). 

In general, if $X$ is a finite set, a \emph{Coxeter matrix} is one whose entries $m_{i,j}\in\mathbb{Z}^+\cup\{\infty\}$ satisfy $m_{i,j}=m_{j,i}$ and $m_{i,j}=1$ if an only if $i=j$ for every $i,j\in X$. It is well known that, up to isomorphism, there is a one-to-one correspondence between Coxeter matrices and Coxeter systems (see~\cite[Theorem 1.1.2]{BjornerBrenti}).

The pancake problem has a straight-forward interpretation in terms of permutations. A stack of $n$ pancakes of different sizes can be thought of as an element of $S_n$ and flipping a stack of pancakes with a spatula can be thought of as using a \emph{prefix reversal permutation}; that is, a permutation whose only action when composed with $w\in S_n$ is to reverse the first so many characters of $w$, in one-line-notation. In other words, using one-line notation, a prefix reversal permutation of $S_n$ has the form 
\begin{align*}&i+1\;\;i\;(i-1)\;\;\cdots 2\;\;1\;\;(i+2)\;\;(i+3)\;\cdots n\\ &=(1,i+1)(2,i)\cdots\left(\left\lfloor\frac{i+2}{2}\right\rfloor,\left\lceil\frac{i+2}{2}\right\rceil\right)\text{, as a product of transpositions},
\end{align*} for some $i\in[n-1]$.

We denote the above permutation by $f_{i}$, with $1\leq i\leq n-1$ and define $P=\{f_1, \ldots, f_{n-1}\}$. For example, in $S_{4}$ one has $f_{1}=2134, f_{2}=3214$, and $f_{3}=4321$. Notice that effect of applying $f_{i}$ to a permutation is similar to that of using a spatula to flip a stack of pancakes since one is reversing the order of the first $i+1$ entries and leaving the rest untouched. 

One can easily see that $s_{i}=f_{i}f_{1}f_{i}$ and that $f_{i}=s_{1}\ldots s_{i-1}\ldots s_{2}s_{1}s_{i}\ldots s_{2}s_{1}$. Hence, $S_{n}$ is also generated by $P$. We refer to the elements of $P$ as \emph{pancake generators} of $S_n$. Furthermore, notice that $f_{i}=s_{1}\ldots s_{i-1}\ldots s_{2}s_{1}s_{i}\ldots s_{2}s_{1}$.

Let $B_n$ be the \emph{hyperoctahedral group}, most commonly referred to as the group of \emph{signed permutations} of the set $[\pm n]=\{-n,-(n-1),\ldots,-1,1,2,\ldots,n\}$. That is, permutations $w$ of $[\pm n]$ satisfying $w(-i)=-w(i)$ for all $i\in[\pm n]$. We shall use \emph{window notation} to denote $w \in B_n$; that is, we denote $w$ by $[w(1), w(2), \ldots, w(n)]$.  The group $B_n$ is generated by the set $\{s^B_0,s^B_1,\ldots, s^B_{n-1}\}$, where $s^B_0=[ \underline{1}\; 2\; \cdots\; n]$ and for $1\leq i \leq n-1$, $s^B_i=[1\;2\;\cdots i+1\;i\; i+2\cdots\; n]$ (see~\cite[Chapter 8]{BjornerBrenti}).

The burnt pancake generators indicate the orientation of the entries: they are negative if they have been reversed an odd number of times and positive otherwise. We define $f^B_i$, $1 \leq i \leq n-1$ to be the signed permutation \begin{align*} f^B_{i}&= [\underline{i+1}\;\underline{i}\;\underline{i-1}\;\cdots \underline{2}\; \underline{1}\; (i+2)\; (i+3)\; \cdots n]\\ &= (1,\,\underline{i+1},\,\underline{1},\, i+1)(2,\, \underline{i},\, \underline{2},\, i)\cdots \left( \left\lfloor\frac{i+2}{2} \right\rfloor, \,\, \underline{\left\lceil\frac{i+2}{2}\right\rceil}, \underline{\left\lfloor\frac{i+2}{2} \right\rfloor}, \,\,  \left\lceil\frac{i+2}{2}\right\rceil \right)\\ &\text{ in disjoint cycle form as elements of the symmetry group of } [\pm n], \end{align*} and $f^B_0=s^B_0$. Thus, for example, in $B_4$ we have $f^B_0=[\underline{1}\;2\;3\;4], f^B_1 =[\underline{2}\;\underline{1}\;3\;4], f^B_2=[\underline{3}\;\underline{2}\;\underline{1}\;4]$, and $f^B_3=[\underline{4}\;\underline{3}\;\underline{2}\;\underline{1}]$. We shall define $P^B=\{f^B_0, f^B_1, \ldots, f^B_{n-1}\}$ as the set of \emph{burnt pancake generators}, or \emph{burnt pancake flips}. It should be noted, that these are the signed versions indicated in this paragraph.

One can see that $s^B_i = f^B_i f^B_0 f^B_1 f^B_0 f^B_i$ for $1\leq i\leq n-1$ and $s^B_0=f^B_0$, thus $B_n$ is also generated by $P^B$. Furthermore, we note that $f^B_i = s^B_0 s^B_1 \ldots s_0^B s^B_{i-1} \ldots s^B_2 s^B_1 s^B_0 s^B_i \ldots s^B_2 s^B_1 s^B_0$.

\section{$S_n$ results}\label{s:snresults}

In this section we take a look at the pancake generators for $S_n$. In this case, the pancake matrix $M_{n-1}=(m_{i,j})_{(n-1)\times (n-1)}$ where $m_{i,j}$ is the order of $f_if_j$ can be derived from \cite[Lemma 1]{KM16}. We include their result and then prove that ``reflections" using the pancake generators are just the set of involutions in $S_n$. We conclude the section with some results for the order of elements of the form $f_if_jf_k$.

The theorem below provides a description for $M_{n-1}$. It turns out $M_{n-1}$ is symmetric and all its diagonal entries are 1.  Most of these entries are described by rephrasing~\cite[Lemma 1]{KM16}.

	\begin{thm}\label{t:main}
	
	If ${m_{i-1,j-1}}$ is the order of $f_{i-1} f_{j-1}$ with $1<i<j\leq n$, then
 		\begin{enumerate}
 		\item\label{c:diagonals} $m_{i-1,i-1}=1$,
 		\item\label{c:symmetry}  $m_{i-1,j-1}=m_{j-1,i-1}$,
 		\item\label{c:2,3} $m_{1,2}=3$, and
 		\item\label{c:lemma1} if $j\geq4$ then
 		\begin{enumerate}
 		\item If $1<i\leq\lfloor\frac{j}{2}\rfloor$ ,then $m_{i-1,j-1}=4$.
 		\item If $1<\lfloor\frac{j}{2}\rfloor<i<j-1$, then
 	\[	m_{i-1,j-1}=\begin{cases}
 		 2q(q+1), &\text{if } r\geq2,t\geq2, \text{ or } r=1,t\geq2,q \text{ is even, }\\ &\text{or } r\geq2,t=1,q \text{ is odd;}\\
 		 q(q+1), &\text{if } r=1,t\geq2,q \text{ is odd, or } r\geq2,t=1,q \text{ is even,}\\ &\text{or } r=1,t=1;\\
 		 2q, &\text{if } r=0.
 		\end{cases} 
 		\] where $d=j-i$, $q=\lfloor\frac{j}{d}\rfloor$, $r=j\pmod{d}$ and $t=d-r$.
 		\item If $i=j-1$, then $m_{i-1,j-1}=j$
 		\end{enumerate}
 		\end{enumerate}
 		\end{thm}
 		
 		\begin{proof}
 		For (\ref{c:diagonals}), since $f_{i}$ is an involution, it follows that $m_{i,i}=1$. 
 		
 		For (\ref{c:symmetry}) notice that $(f_{i}f_{j})^{-1} = f_{j}f_{i}$, so it follows that $m_{i,j}=m_{j,i}$, and therefore $M_{n-1}$ is symmetric. 
 		
 		Case (\ref{c:2,3}) Follows from direct computation. 
 		
 		For Case (\ref{c:lemma1}), notice that elements in $S_{n-1}$ can be viewed as elements in $S_n$ leaving $n$ fixed, the matrix $M_{n-1}$ can be viewed as a submatrix of $M_n$ by ignoring the last row and column of $M_n$. So Case (\ref{c:lemma1}) follows from~\cite[Lemma 1]{KM16} by having $n$ take different values. 
 		\end{proof}
 		
 		\begin{rem}
 		     We point out that in~\cite{KM16}, the authors use $r_j$ with $2\leq j\leq n$ to denote the permutation that reverses the first $j$ terms from the identity permutation $123\ldots n$. In other words, $r_j=f_{j-1}$, for $2\leq j\leq n$. However, our notation resembles the notation that is used for $S_n$ viewed as a Coxeter group generated by $S$, the set of adjacent transpositions. 
 		   		\end{rem}

	\begin{ex} In Figure~\ref{f:matrix}, we depict $19\times 19$ Coxeter matrix $M_{19}$ for $S_{20}$.
	\end{ex}
	
		We now describe the set of ``reflections'' with respect to $P$, that is, the conjugates of elements of $P$ by permutations. 
	
	\begin{thm}\label{t:reflec}
		If $T=\{w f_i w^{-1}\,|\, i \in [n], w \in S_n\}$, the set of conjugates of the pancake generators, then $T$ is the set of all involutions in $S_n$.
	\end{thm}
	\begin{proof}
	    If $f_i\in P$ and $w\in S_n$, then $(wf_iw^{-1})^2=e$, so each element in $T$ is an involution. 
	
		Conversely, suppose $t$ is an arbitrary involution in $S_n$. Then the $t$ can be written in disjoint cycle notation using only length two cycles. Say $t=(a_1, b_1)(a_2, b_2)\cdots(a_k, b_k)$ with $a_1 < a_2 < \cdots < a_k$ and $a_i < b_i$ for all $i \in [k]$. We know that $k \leq \left\lfloor \dfrac{n}{2} \right\rfloor$ thus $2k-1 \leq 2 \left\lfloor \dfrac{n}{2} \right\rfloor -1 \leq n-1$. Consider the flip $$f_{2k-1} = (1,2k)(2,2k-1)\cdots(k,k+1)$$ which consists of $k$ disjoint two-cycles and $$w=a_1 a_2 \ldots a_k b_k b_{k-1} \ldots b_2 b_1 w_{2k+1} \ldots w_n\text{, in one-line notation,}$$  where $w_{2k+1} \ldots w_n$ is an arbitrary permutation of $[n] \setminus \{a_1, b_1, a_2, b_2, \ldots, a_k, b_k\}$.
		
		The element $t$ is in $T$ if $ wf_{2k-1} = tw$. Notice that $$ wf_{2k-1} = b_1 b_2 \ldots b_{k-1} b_k a_k a_{k-1} \ldots a_2 a_1 w_{2k+1} \ldots w_n.$$ Furthermore, $$ tw = b_1 b_2 \ldots b_{k-1} b_k a_k a_{k-1} \ldots a_2 a_1 w_{2k+1} \ldots w_n.$$ Therefore $t \in T$. 
	\end{proof}
	
	Since $|T|$ is the same as the number of involutions in $S_n$, we have the following corollary. 
	
	\begin{cor}\label{c:reflecOrder}
		$\left|T\right| = \displaystyle\sum\limits_{k=1}^{\lfloor n/2 \rfloor} \dfrac{n!}{2^k(n-2k)!k!}$.
	\end{cor}
	
	\begin{rem}
	In the symmetric group, every reflection, that is, every element of the form $wsw^{-1}$, where $s$ is an adjacent transposition and $w$ is a permutation, is an involution. However, there are involutions in the symmetric group that are not reflections. As Theorem~\ref{t:reflec} shows, if we use the pancake generators for the symmetric group, the ``reflections'' obtained are indeed \textbf{all} the involutions in $S_{n}$. 
	\end{rem}
	
\begin{center}
\begin{figure}[H]
	\begin{tikzpicture}
	 	\matrix[matrix of math nodes, left delimiter=(, right delimiter=), row sep=1pt, column sep=1pt] (m){ 
		1  & 3  & 4  & 4  & 4  & 4  & 4  & 4  & 4  & 4  & 4  & 4  & 4  & 4  & 4  & 4  & 4  & 4  & 4 \\ 
		3  & 1  & 4  & 6  & 4  & 4  & 4  & 4  & 4  & 4  & 4  & 4  & 4  & 4  & 4  & 4  & 4  & 4  & 4 \\ 
		4  & 4  & 1  & 5  & 6  & 12 & 4  & 4  & 4  & 4  & 4  & 4  & 4  & 4  & 4  & 4  & 4  & 4  & 4 \\ 
		4  & 6  & 5  & 1  & 6  & 12 & 6  & 12 & 4  & 4  & 4  & 4  & 4  & 4  & 4  & 4  & 4  & 4  & 4 \\ 
		4  & 4  & 6  & 6  & 1  & 7  & 8  & 6  & 12 & 12 & 4  & 4  & 4  & 4  & 4  & 4  & 4  & 4  & 4 \\ 
		4  & 4  & 12 & 12 & 7  & 1  & 8  & 20 & 12 & 6  & 12 & 12 & 4  & 4  & 4  & 4  & 4  & 4  & 4 \\ 
		4  & 4  & 4  & 6  & 8  & 8  & 1  & 9  & 10 & 24 & 6  & 12 & 12 & 12 & 4  & 4  & 4  & 4  & 4 \\ 
		4  & 4  & 4  & 12 & 6  & 20 & 9  & 1  & 10 & 30 & 8  & 12 & 6  & 12 & 12 & 12 & 4  & 4  & 4 \\ 		
		4  & 4  & 4  & 4  & 12 & 12 & 10 & 10 & 1  & 11 & 12 & 40 & 24 & 6  & 12 & 12 & 12 & 12 & 4 \\ 
		4  & 4  & 4  & 4  & 12 & 6  & 24 & 30 & 11 & 1  & 12 & 42 & 20 & 24 & 12 & 6  & 12 & 12 & 12\\ 
		4  & 4  & 4  & 4  & 4  & 12 & 6  & 8  & 12 & 12 & 1  & 13 & 14 & 10 & 8  & 24 & 6  & 12 & 12\\ 
		4  & 4  & 4  & 4  & 4  & 12 & 12 & 12 & 40 & 42 & 13 & 1  & 14 & 56 & 30 & 40 & 24 & 12 & 6 \\ 
		4  & 4  & 4  & 4  & 4  & 4  & 12 & 6  & 24 & 20 & 14 & 14 & 1  & 15 & 16 & 60 & 40 & 24 & 24\\ 
		4  & 4  & 4  & 4  & 4  & 4  & 12 & 12 & 6  & 24 & 10 & 56 & 15 & 1  & 16 & 72 & 12 & 20 & 8 \\ 
		4  & 4  & 4  & 4  & 4  & 4  & 4  & 12 & 12 & 12 & 8  & 30 & 16 & 16 & 1  & 17 & 18 & 84 & 10\\ 
		4  & 4  & 4  & 4  & 4  & 4  & 4  & 12 & 12 & 6  & 24 & 40 & 60 & 72 & 17 & 1  & 18 & 90 & 42\\ 
		4  & 4  & 4  & 4  & 4  & 4  & 4  & 4  & 12 & 12 & 6  & 24 & 40 & 12 & 18 & 18 & 1  & 19 & 20\\ 
		4  & 4  & 4  & 4  & 4  & 4  & 4  & 4  & 12 & 12 & 12 & 12 & 24 & 20 & 84 & 90 & 19 & 1  & 20\\ 
		4  & 4  & 4  & 4  & 4  & 4  & 4  & 4  & 4  & 12 & 12 & 6  & 24 & 8  & 10 & 42 & 20 & 20 & 1  \\
		};
		\end{tikzpicture}
	\caption{Pancake Matrix $M_{19}$ for $S_{20}$. Notice that the matrix is symmetric, the entries in the main diagonal are all 1 and the entries in the off-diagonal are the positive integers that are at least 3.}
	\label{f:matrix}
\end{figure}
\end{center}

	\subsection{Order of $f_{i}f_{j}f_{k}$}
	
	We now discuss the order $m_{i,j,k}$ of $f_{i}f_{j}f_{k}$. It turns out that all we need is to understand the order in the case $i\leq j\leq k$, as the order of $f_{\sigma(i)}f_{\sigma(j)}f_{\sigma(k)}$ is also $m_{i,j,k}$, as shown in the following lemma. 
	
	\begin{lem}\label{l:rank3permute}
	For all $i,j,k$ with $1\leq i,j,k\leq n$ and any permutation $\sigma$ of $\{i,j,k\}$, the order of $f_{i}f_{j}f_{k}$ is the same as the order of $f_{\sigma(i)}f_{\sigma(j)}f_{\sigma(k)}$.
	\end{lem}
	\begin{proof}
	There are two cases to consider.
	\begin{description}
	\item[Case $|\{i,j,k\}|< 3$] In this case, the order of $f_{i}f_{j}f_{k}$ is 2. Indeed, if $|\{i,j,k\}|=2$ then $f_{i}f_{j}f_{k}$ has the form 
	$f_{a}f_{a}f_{b}$ or $f_{a}f_{b}f_{a}$ or $f_{b}f_{a}f_{a}$, for $a,b\in\{i,j,k\}$, all of which have order two. Furthermore, if $|\{i,j,k\}|=1$, then $f_{i}f_{i}f_{i}=f_{i}$, which also has order 2. 
	\item[Case $|\{i,j,k\}|= 3$] Notice that $f_{i}f_{j}f_{k}, f_{j}f_{k}f_{i}$, and  $f_{k}f_{i}f_{j}$ are in the same conjugacy class of $S_{n}$; for example, $f_{k}f_{i}f_{j}=f_{k}(f_{i}f_{j}f_{k})f_{k}$. Therefore they all have the same order as they have the same cycle structure. Moreover, $f_{k}f_{j}f_{i}, f_{i}f_{k}f_{j},$ and $f_{j}f_{i}f_{k}$ are in the same conjugacy class, and so they have the same order as well. Since $f_{k}f_{j}f_{i}=(f_{i}f_{j}f_{k})^{-1}$, the lemma follows.
		\end{description}
	\end{proof}
	
	Here are a collection of partial results on the orders of elements of the form $f_{i} f_{j} f_{k}$. Specifically these are all of the relations where the leftmost generator is $f_1$, i.e. all of the orders of $f_1 f_j f_k$.

    \begin{thm}\label{t:rank3}
    If $m_{1,j-1,k-1}$ is the order of $f_1f_{j-1}f_{k-1}$ with $1<j<k\leq n$, then 
        \begin{enumerate}
            \item\label{i:3involutions} $m_{1,1,j-1}=m_{1,j-1,j-1}=2$,
            
            \item\label{i:sixes} if $j \geq 6$ then $m_{1,2,j-1}=6$,
            
            \item\label{i:diff1} if $j=k-1$ then $m_{1,j-1,k-1}=k-1$,
            
            \item\label{i:diff2} if $j=k-2$ and $k$ is odd or $j=k-3$ and $2 \neq k \pmod{3}$ then $m_{1,j-1,k-1}=k$,
            
            \item\label{i:lcmcases} if $k \geq 5$ then
                \begin{enumerate}
                    \item\label{i:rlessthan2} $m_{1,j-1,k-1}=
                        \begin{cases} 
                            4q, &\text{if } r=0, d \geq 4;\\
                            2q+1, &\text{if } r=1, d=2;\\
                            q(3q+1), &\text{if } r=1, d=4 \text{ or } r=1, d \geq 5, q \text{ is odd};\\
                            2q(3q+1), &\text{if } r=1, d \geq 5, q \text{ is even}.
                        \end{cases}$
                    \item\label{i:rgreaterthan2} $m_{1,j-1,k-1}=
                        \begin{cases}
                            q(q+1), &\text{if } r=2,d=3, \text{ or } r=2, d \geq 4, q \text{ is odd}, \text{ or }\\& r=3, d=4, 0=q\pmod{3}, \text{ or } r=3, d \geq 5,\\& q=3\pmod{6}, \text{ or } r \geq 4, d \geq 5, 0= q\pmod{4};\\
                            2q(q+1), &\text{if } r=2, d \geq 4, q \text{ is even}, \text{ or } r=3, d \geq 5,\\& 0=q\pmod{6}, \text{ or } r \geq 4, d \geq 5, 2=q\pmod{4};\\
                            3q(q+1), &\text{if } r=3, d=4, 0 \neq q\pmod{3}, \text{ or } r=3, d \geq 5,\\& 1,5=q\pmod{6};\\
                            4q(q+1), &\text{if } r \geq 4, d \geq 5, q \text{ is odd};\\
                            6q(q+1), &\text{if } r=3, d \geq 5,\, 2,4=q\pmod{6};
                        \end{cases}$
                \end{enumerate}
                where $d=k-j, q=\left\lfloor\frac{k}{d}\right\rfloor,$ and $r=k\pmod{d}.$
        \end{enumerate}
    \end{thm}
    
    \begin{proof}
        For Case (\ref{i:3involutions}), note that for any $1 \leq i \leq n-1$, $(f_i)^2=e$. So $f_1f_1f_{j-1}=f_{j-1}$, which is order two, and $f_1f_{j-1}f_{j-1}=f_1$, which is also order two.
        
        For each of the following cases we will look at the disjoint cycle notation of the permutations to find the order of the three generators. 
        
        For Case (\ref{i:sixes}), let $j \geq 6$.
            \begin{align*}
                f_1 f_2 f_{j-1} &= \left( \begin{matrix} 1 & 2 & 3 & 4 & \cdots & j\\ j-1 & j-2 & j & j-3 & \cdots & 1 \end{matrix} \right)\\ &= \left( 1, j-1, 2, j-2, 3, j \right) (4, j-3) \ldots \left( \left\lfloor \dfrac{j+1}{2} \right\rfloor, \left\lceil \dfrac{j+1}{2} \right\rceil \right).         
            \end{align*}
        The least common multiple of these lengths is 6, which is the order of the permutation.
        
        Since elements in $S_n$ is a parabolic subgroup of $S_{n+1}$, generated by all but the largest indexed generator, then the matrix of $\left(m_{1,j-1,k-1}\right)_{1 \leq j,k \leq n}$ is a submatrix of the matrix $\left(m_{1,j-1,k-1}\right)_{1 \leq j,k \leq n+1}$ with the last row and last column removed. Thus it is sufficient to consider only the cases with $k=n$.
        
        For Case (\ref{i:diff1}), the three generators result in
            \begin{align*}
                f_1 f_{j-1} f_{n-1} &= \left( \begin{matrix} 1 & 2 & 3 & \cdots & n-1 & n\\ 3 & 2 & 4 & \cdots & n & 1 \end{matrix} \right)\\
                &= (1, 3, 4, \ldots, n). 
            \end{align*}
        The length of this disjoint cycle is $n-1$. Thus the order is $n-1$.
        
        For Case (\ref{i:diff2}), first consider $j=n-2$ and $n$ is odd. 
            \begin{align*}
                f_1 f_{j-1} f_{n-1} &= \left( \begin{matrix} 1 & 2 & 3 & \cdots & n-2 & n-1 & n\\ 4 & 3 & 5 & \cdots & n & 2 & 1 \end{matrix} \right)\\
                &= (1, 4, 6, \ldots, n-1, 2, 3, 5, \ldots, n),
            \end{align*}
        whose cycle length is $n$. Second consider $j=n-3$ and $2 \neq n\pmod{3}$. Say $n=3\mathfrak{q} + \mathfrak{r}$ with $\mathfrak{r} = 0$ or 1.
            \begin{align*}
                f_1 f_{j-1} f_{n-1} &= \left( \begin{matrix} 1 & 2 & 3 & \cdots & n-3 &  n-2 & n-1 & n\\ 5 & 4 & 6 & \cdots & n & 3 & 2 & 1 \end{matrix} \right)\\
                &= (1, 5, 8, \ldots, 3(\mathfrak{q}-1)+2, 2, 4, 7, \ldots, 3(\mathfrak{q}-1)+1, 3, 6, \ldots, 3\mathfrak{q}), \text{ when } \mathfrak{r}=0.\\
                f_1 f_{j-1} f_{n-1} &= (1, 5, 8, \ldots, 3(\mathfrak{q}-1)+2, 3, 6, 9, \ldots, 3\mathfrak{q}, 2, 4, 7, \ldots 3\mathfrak{q}+1), \text{ when } \mathfrak{r}=1.
            \end{align*}
        With both possible $\mathfrak{r}$ the length of the disjoint cycle is $n$.
        
        For Case (\ref{i:lcmcases}) let $d=n-j$, $q=\left\lfloor\frac{n}{d}\right\rfloor$, and $r=n\pmod{d}$. We will consider each possible value of $r$ followed by the possibilities of $d$. Since $n=qd+r$ and $d=n-j$, then $j=(q-1)d+r$. In general, three generators form the permutation 
            \begin{equation}\label{e:3gens}
                f_1 f_{j-1} f_{n-1} = \left( \begin{matrix} 1 & 2 & 3 & \cdots & (q-1)d+r & (q-1)d+r+1 & \cdots & qd+r\\ d+2 & d+1 & d+3 & \cdots & qd+r & d & \cdots & 1 \end{matrix} \right)
            \end{equation}
        
        Say $r=0$, then in disjoint cycle notation we have that 
            \begin{align*}
                (\ref{e:3gens}) &= (1, d+2, 2d+2, \ldots, (q-1)d+2, d-1, 2d-1, \ldots, qd-1,\\ &\qquad 2, d+1, 2d+1, \ldots, (q-1)d+1, d, 2d, \ldots, qd)\\
                &\quad (3, d+3, 2d+3, \ldots, (q-1)d+3, d-2, 2d-2, \ldots, qd-2)\\
                &\quad (4, d+4, 2d+4, \ldots, (q-1)d+4, d-3, 2d-3, \ldots, qd-3)\\
                &\qquad \vdots\\
                &\quad \Bigg(\left\lfloor\dfrac{d}{2}\right\rfloor, d+\left\lfloor\dfrac{d}{2}\right\rfloor, \ldots, (q-1)d+\left\lfloor\dfrac{d}{2}\right\rfloor, \left\lceil\dfrac{d}{2}\right\rceil+1, d+\left\lceil\dfrac{d}{2}\right\rceil+1, \\
                &\qquad \ldots,(q-1)d+\left\lceil\dfrac{d}{2}\right\rceil+1\Bigg)\\
                &\quad \left(\left\lfloor\dfrac{d}{2}\right\rfloor+1, d+\left\lfloor\dfrac{d}{2}\right\rfloor+1, \ldots, (q-1)d+\left\lfloor\dfrac{d}{2}\right\rfloor+1\right).
            \end{align*}
        The first cycle is of length $4q$, the next $\left\lfloor\frac{d}{2}\right\rfloor - 2$ cycles are of length $2q$, and the last cycle is length $q$. The least common multiple of these lengths is $4q$, the order of the permutation. To verify that these are all the disjoint cycles we can see that the number of characters affected is $$4q+2q\left(\left\lfloor\frac{d}{2}\right\rfloor-2\right)+q=4q+q(d-1-4)+q=qd=n.$$
        
        Say $r=1$ and $d=2$, then in disjoint cycle notation we have that 
            \begin{align*}
                (\ref{e:3gens}) &= (1,4,6,\ldots,2q,2,3,5,\ldots,2q+1).
            \end{align*}
        All of the characters are part of this cycle, since $n=2q+1$, and therefore the order of the permutation is $2q+1$.
        
        Say $r=1$ and $d=4$, then the disjoint cycle notation is
            \begin{align*}
                (\ref{e:3gens}) &= (1, 6, 10, \ldots, 4q-2, 4, 8, \ldots, 4q, 2, 5, 9, \ldots, 4q+1)\\
                &\quad (3, 7, \ldots, 4q-1).
            \end{align*}
        The length of the first cycle is $3q+1$ and the second cycle is of length $q$. The least common multiple of the cycle lengths is $q(3q+1)$. All of the characters are part of this cycle, since $n=4q+1$.
        
        Say $r=1$ and $d \geq 5$. The disjoint cycles are
            \begin{align*}
                (\ref{e:3gens}) &= (1, d+2, 2d+2, \ldots, (q-1)d+2, d, 2d, \ldots, qd, 2, d+1, 2d+1, \ldots, qd+1)\\
                &\quad \left(3, d+3, \ldots, (q-1)d+3, d-1, 2d-1, \ldots, qd-1 \right)\\
                &\quad \left(4, d+4, \ldots, (q-1)d+4, d-2, 2d-2, \ldots, qd-2 \right)\\
                &\qquad \vdots\\
                &\quad \Bigg(\left\lfloor\dfrac{d+2}{2}\right\rfloor, d+\left\lfloor\dfrac{d+2}{2}\right\rfloor, \ldots, (q-1)d+\left\lfloor\dfrac{d+2}{2}\right\rfloor, \left\lceil\dfrac{d+2}{2}\right\rceil, d+\left\lceil\dfrac{d+2}{2}\right\rceil, \\
                &\qquad \ldots, (q-1)d+\left\lceil\dfrac{d+2}{2}\right\rceil \Bigg).
            \end{align*}
        The first cycle is of length $3q+1$ and the other $\left\lfloor\frac{d+2}{2}\right\rfloor-2$ cycles are length $2q$. When $q$ is even the least common multiple of these lengths is $2q(3q+1)$. When $q$ is odd, $3q+1$ is even, and thus the least common multiple of the lengths is $q(3q+1)$. All of the disjoint cycles are accounted for since the total number of characters in the cycles is $$3q+1 + 2q\left(\left\lfloor\dfrac{d+2}{2}\right\rfloor -2\right) = 3q+1+q(d+1-4)=qd+1.$$ 
        
        Say $r=2$ and $d=3$, then the disjoint cycles are 
            \begin{align*}
                (\ref{e:3gens}) &= (1,5,8,\ldots,3q+2)\,\,(2,4,7,\ldots,3q+1)\,\,(3,6,\ldots,3q).
            \end{align*}
        The first two cycles are of length $q+1$ and the last cycle is length $q$. The least common multiple of these lengths is $q(q+1)$. It is also clear that these are all the cycles since the sum of the cycle lengths is $3q+2=n$.
        
        Say $r=2$ and $d \geq 4$. The disjoint cycles are 
            \begin{align*}
                (\ref{e:3gens}) &= (1, d+2, 2d+2, \ldots, qd+2) \,\,(2, d+1, 2d+1, \ldots, qd+1)\\
                &\quad (3, d+3, \ldots, (q-1)d+3, d, 2d, \ldots, qd)\\
                &\quad (4, d+4, \ldots, (q-1)d+4, d-1, 2d-1, \ldots, qd-1)\\
                &\qquad \vdots\\
                &\quad \Bigg(\left\lfloor\dfrac{d+3}{2}\right\rfloor, d+\left\lfloor\dfrac{d+3}{2}\right\rfloor, \ldots, (q-1)d+\left\lfloor\dfrac{d+3}{2}\right\rfloor, \left\lceil\dfrac{d+3}{2}\right\rceil, d+\left\lceil\dfrac{d+3}{2}\right\rceil,\\
                &\qquad\ldots, (q-1)d+\left\lceil\dfrac{d+3}{2}\right\rceil\Bigg).
            \end{align*}
        The first two cycles are length $q+1$ and the last $\left\lfloor\frac{d+3}{2}\right\rfloor-2$ cycles are of length $2q$. When $q$ is odd then the least common multiple of the lengths is $q(q+1)$. When $q$ is even then the least common multiple of the lengths is $2q(q+1)$. These are all of the disjoint cycles since the number of characters in them is $$2q+2+2q\left(\left\lfloor\dfrac{d+3}{2}\right\rfloor-2\right)=2q+2+q(d+2-4)=qd+2=n.$$
        
        Say $r=3$ and $d=4$. The disjoint cycles are
            \begin{align*}
                (\ref{e:3gens}) &= (1,6,10,\ldots,4q+2,2,5,9,\ldots,4q+1,3,7,\ldots,4q+3)\\
                &\quad (4,8,\ldots, 4q).
            \end{align*}
        The first cycle is of length $3q+3$ and the second cycle is of length $q$. If $q$ is a multiple of $3$, then the least common multiple is $q(q+1)$. If $q$ is not a multiple of $3$, then the least common multiple is $3q(q+1)$. The number of characters in both cycles is $4q+3=n$.
        
        Say $r=3$ and $d \geq 5$. The disjoint cycles are 
            \begin{align*}
                (\ref{e:3gens}) &= (1, d+2, 2d+2, \ldots, qd+2, 2, d+1, 2d+1, \ldots, qd+1, 3, d+3, \ldots, qd+3)\\
                &\quad (4, d+4, \ldots, (q-1)d+4, d, 2d, \ldots, qd)\\
                &\quad (5, d+5, \ldots, (q-1)d+5, d-1, 2d-1, \ldots, qd-1)\\
                &\qquad \vdots \\
                &\quad \Bigg(\left\lfloor\dfrac{d+4}{2}\right\rfloor, d+\left\lfloor\dfrac{d+4}{2}\right\rfloor, \ldots, (q-1)d+\left\lfloor\dfrac{d+4}{2}\right\rfloor, \left\lceil\dfrac{d+4}{2}\right\rceil, d+\left\lceil\dfrac{d+4}{2}\right\rceil,\\
                &\qquad \ldots, (q-1)d+\left\lceil\dfrac{d+4}{2}\right\rceil\Bigg).
            \end{align*}
        The first cycle is of length $3q+3$ and the remaining $\left\lfloor\frac{d+4}{2}\right\rfloor-3$ cycles are of length $2q$. When $q$ is odd and divisible by $3$ the least common multiple is $q(q+1)$. When $q$ divisible by $6$ the least common multiple is $2q(q+1)$. When $q$ is odd and not divisible by $3$ the least common multiple is $3q(q+1)$. When $q$ is even but not divisible by $3$ the least common multiple is $6q(q+1)$. The number of characters in all of the cycles is $$3q+3+2q\left(\left\lfloor\dfrac{d+4}{2}\right\rfloor-3\right)=3q+3+q(d+3-6)=qd+3=n.$$
        
        Say $r=4$ and $d=5$. The disjoint cycles are
            \begin{align*}
                (\ref{e:3gens}) &= (1,7,12,\ldots,5q+2,3,8,13,\ldots,5q+3,2,6,11,\ldots,5q+1,4,9,\ldots,5q+4)\\
                &\quad (5,10,\ldots, 5q).
            \end{align*}
        The first cycle is of length $4q+4$ and the second cycle is of length $q$. If $q$ is a multiple of $4$, then the least common multiple is $q(q+1)$. If $q$ is even but not a multiple of $4$, then the least common multiple is $2q(q+1)$. If $q$ is odd, then the least common multiple is $4q(q+1)$. The number of characters in both cycles is $5q+4=n$.
        
        Say $r=4$ and $d \geq 6$. The disjoint cycles are 
            \begin{align*}
                (\ref{e:3gens}) &= (1, d+2, 2d+2, \ldots, qd+2, 3, d+3, \ldots, qd+3, 2, d+1, 2d+1,\\
                &\qquad \ldots, qd+1, 4, d+4, \ldots, qd+4)\\
                &\quad (5, d+5, \ldots, (q-1)d+5, d, 2d, \ldots, qd)\\
                &\quad (6, d+6, \ldots, (q-1)d+6, d-1, 2d-1, \ldots, qd-1)\\
                &\qquad \vdots \\
                &\quad \Bigg(\left\lfloor\dfrac{d+5}{2}\right\rfloor, d+\left\lfloor\dfrac{d+5}{2}\right\rfloor, \ldots, (q-1)d+\left\lfloor\dfrac{d+5}{2}\right\rfloor, \left\lceil\dfrac{d+5}{2}\right\rceil, d+\left\lceil\dfrac{d+5}{2}\right\rceil,\\
                &\qquad \ldots, (q-1)d+\left\lceil\dfrac{d+5}{2}\right\rceil\Bigg).
            \end{align*}
        The first cycle is of length $4q+4$ and the remaining $\left\lfloor\frac{d+5}{2}\right\rfloor-4$ cycles are of length $2q$. When $q$ is even and not divisible by $4$ the least common multiple is $2q(q+1)$. When $q$ is odd the least common multiple is $4q(q+1)$. The number of characters in all of the cycles is $$4q+4+2q\left(\left\lfloor\dfrac{d+5}{2}\right\rfloor-4\right)=4q+4+q(d+4-8)=qd+4=n.$$
        
        Say $r \geq 5$ and $d \geq 6$. The disjoint cycles are 
            \begin{align*}
                (\ref{e:3gens}) &= (1, d+2, 2d+2, \ldots, qd+2, r-1, d+r-1, \ldots, qd+r-1, \\
                &\qquad 2, d+1, 2d+1, \ldots, qd+1, r, d+r, \ldots, qd+r)\\
                &\quad (3, d+3, \ldots, qd+3, r-2, d+r-2, \ldots, qd+r-2)\\
                &\quad (4, d+4, \ldots, qd+4, r-3, d+r-3, \ldots, qd+r-3)\\
                &\qquad \vdots \\
                &\quad \Bigg(\left\lfloor\dfrac{r+1}{2}\right\rfloor, d+\left\lfloor\dfrac{r+1}{2}\right\rfloor, \ldots, qd+\left\lfloor\dfrac{r+1}{2}\right\rfloor, \left\lceil\dfrac{r+1}{2}\right\rceil, d+\left\lceil\dfrac{r+1}{2}\right\rceil,\\
                &\qquad \ldots, qd+\left\lceil\dfrac{r+1}{2}\right\rceil\Bigg)\\
                &\quad (r+1, d+r+1, \ldots, (q-1)d+r+1, d, 2d, \ldots, qd)\\
                &\quad (r+2, d+r+2, \ldots, (q-1)d+r+2, d-1, 2d-1, \ldots, qd-1)\\
                &\qquad \vdots\\
                &\quad \Bigg(\left\lfloor\dfrac{d+r+1}{2}\right\rfloor, d+\left\lfloor\dfrac{d+r+1}{2}\right\rfloor, \ldots, (q-1)d+\left\lfloor\dfrac{d+r+1}{2}\right\rfloor, \left\lceil\dfrac{d+r+1}{2}\right\rceil,\\
                &\qquad d+\left\lceil\dfrac{d+r+1}{2}\right\rceil, \ldots, (q-1)d+\left\lceil\dfrac{d+r+1}{2}\right\rceil\Bigg).
            \end{align*}
        The first cycle is of length $4q+4$, the next $\left\lfloor\frac{r+1}{2}\right\rfloor-2$ cycles are each length $2q+2$, and the last $\left\lfloor\frac{d+r+1}{2}\right\rfloor-r$ cycles are of length $2q$. When $q$ is even and not divisible by $4$ the least common multiple is $2q(q+1)$. When $q$ is odd the least common multiple is $4q(q+1)$. The number of characters in all of the cycles is 
            \begin{align*}
                4q&+4+2(q+1)\left(\left\lfloor\dfrac{r+1}{2}\right\rfloor-2\right)+2q\left(\left\lfloor\dfrac{d+r+1}{2}\right\rfloor-r\right)\\
                &=4q+4+(q+1)(r-4)+q(d+r-2r) = qd+r=n.
            \end{align*}
    \end{proof}
	
	   \begin{rem}
	   Although Theorem~\ref{t:rank3} and Lemma~\ref{l:rank3permute} fully characterize the order of elements of the form $f_1f_jf_k$, we do not cover all possible combinations of three generators $f_if_jf_k$. In order to record all such orders we propose a generalization of the ``Coxeter matrix'' which we refer to as the \emph{Coxeter tensor} of rank $r$. In essence it would be a matrix of matrices or vectors depending on the parity of the rank. For three generators the object recording all of the orders would be the Coxeter 3-tensor, visualized as a cube. Using Lemma \ref{l:rank3permute} we intend to look at the order of the principal tetrahedron consisting of only entries with increasing indices. Theorem \ref{t:rank3} would be the base of said tetrahedron.
	   \end{rem}
	    
		\begin{ex} In Figure~\ref{f:3matrix}, we depict $24\times 24$ matrix whose $(j,k)$ entry is the order of $f_{1}f_{j}f_{k}$ in $S_{25}$.
\begin{center}
\begin{figure}[H]
	\begin{tikzpicture}[every node/.style={font=\footnotesize}]
	 	\matrix[matrix of math nodes, left delimiter=(, right delimiter=), row sep=0.15pt, column sep=0.15pt] (m){ 
			2 & 2 & 2 & 2 & 2 & 2 & 2 & 2 & 2 & 2 & 2 & 2 & 2 & 2 & 2 & 2 & 2 & 2 & 2 & 2 & 2 & 2 & 2 & 2\\ 
            2 & 2 & 3 & 5 & 6 & 6 & 6 & 6 & 6 & 6 & 6 & 6 & 6 & 6 & 6 & 6 & 6 & 6 & 6 & 6 & 6 & 6 & 6 & 6\\ 
            2 & 3 & 2 & 4 & 3 & 7 & 8 & 8 & 8 & 8 & 8 & 8 & 8 & 8 & 8 & 8 & 8 & 8 & 8 & 8 & 8 & 8 & 8 & 8\\ 
            2 & 5 & 4 & 2 & 5 & 7 & 6 & 14 & 8 & 8 & 8 & 8 & 8 & 8 & 8 & 8 & 8 & 8 & 8 & 8 & 8 & 8 & 8 & 8\\ 
            2 & 6 & 3 & 5 & 2 & 6 & 4 & 9 & 12 & 28 & 8 & 8 & 8 & 8 & 8 & 8 & 8 & 8 & 8 & 8 & 8 & 8 & 8 & 8\\ 
            2 & 6 & 7 & 7 & 6 & 2 & 7 & 9 & 10 & 18 & 12 & 28 & 8 & 8 & 8 & 8 & 8 & 8 & 8 & 8 & 8 & 8 & 8 & 8\\ 
            2 & 6 & 8 & 6 & 4 & 7 & 2 & 8 & 5 & 12 & 12 & 36 & 12 & 28 & 8 & 8 & 8 & 8 & 8 & 8 & 8 & 8 & 8 & 8\\ 
            2 & 6 & 8 & 14 & 9 & 9 & 8 & 2 & 9 & 11 & 12 & 30 & 12 & 36 & 12 & 28 & 8 & 8 & 8 & 8 & 8 & 8 & 8 & 8\\
            2 & 6 & 8 & 8 & 12 & 10 & 5 & 9 & 2 & 10 & 6 & 13 & 12 & 12 & 12 & 36 & 12 & 28 & 8 & 8 & 8 & 8 & 8 & 8\\ 
            2 & 6 & 8 & 8 & 28 & 18 & 12 & 11 & 10 & 2 & 11 & 13 & 20 & 12 & 30 & 12 & 12 & 36 & 12 & 28 & 8 & 8 & 8 & 8\\ 
            2 & 6 & 8 & 8 & 8 & 12 & 12 & 12 & 6 & 11 & 2 & 12 & 7 & 15 & 16 & 12 & 12 & 12 & 12 & 36 & 12 & 28 & 8 & 8\\ 
            2 & 6 & 8 & 8 & 8 & 28 & 36 & 30 & 13 & 13 & 12 & 2 & 13 & 15 & 16 & 52 & 12 & 30 & 12 & 12 & 12 & 36 & 12 & 28\\ 
            2 & 6 & 8 & 8 & 8 & 8 & 12 & 12 & 12 & 20 & 7 & 13 & 2 & 14 & 8 & 30 & 40 & 48 & 12 & 12 & 12 & 12 & 12 & 36\\ 
            2 & 6 & 8 & 8 & 8 & 8 & 28 & 36 & 12 & 12 & 15 & 15 & 14 & 2 & 15 & 17 & 18 & 60 & 16 & 12 & 30 & 12 & 12 & 12\\ 
            2 & 6 & 8 & 8 & 8 & 8 & 8 & 12 & 12 & 30 & 16 & 16 & 8 & 15 & 2 & 16 & 9 & 19 & 20 & 104 & 48 & 12 & 12 & 12\\ 
            2 & 6 & 8 & 8 & 8 & 8 & 8 & 28 & 36 & 12 & 12 & 52 & 30 & 17 & 16 & 2 & 17 & 19 & 42 & 80 & 40 & 48 & 12 & 30\\ 
            2 & 6 & 8 & 8 & 8 & 8 & 8 & 8 & 12 & 12 & 12 & 12 & 40 & 18 & 9 & 17 & 2 & 18 & 10 & 21 & 30 & 120 & 16 & 48\\ 
            2 & 6 & 8 & 8 & 8 & 8 & 8 & 8 & 28 & 36 & 12 & 30 & 48 & 60 & 19 & 19 & 18 & 2 & 19 & 21 & 22 & 90 & 20 & 104\\ 
            2 & 6 & 8 & 8 & 8 & 8 & 8 & 8 & 8 & 12 & 12 & 12 & 12 & 16 & 20 & 42 & 10 & 19 & 2 & 20 & 11 & 56 & 24 & 20\\ 
            2 & 6 & 8 & 8 & 8 & 8 & 8 & 8 & 8 & 28 & 36 & 12 & 12 & 12 & 104 & 80 & 21 & 21 & 20 & 2 & 21 & 23 & 24 & 114\\ 
            2 & 6 & 8 & 8 & 8 & 8 & 8 & 8 & 8 & 8 & 12 & 12 & 12 & 30 & 48 & 40 & 30 & 22 & 11 & 21 & 2 & 22 & 12 & 25\\ 
            2 & 6 & 8 & 8 & 8 & 8 & 8 & 8 & 8 & 8 & 28 & 36 & 12 & 12 & 12 & 48 & 120 & 90 & 56 & 23 & 22 & 2 & 23 & 25\\ 
            2 & 6 & 8 & 8 & 8 & 8 & 8 & 8 & 8 & 8 & 8 & 12 & 12 & 12 & 12 & 12 & 16 & 20 & 24 & 24 & 12 & 23 & 2 & 24\\ 
            2 & 6 & 8 & 8 & 8 & 8 & 8 & 8 & 8 & 8 & 8 & 28 & 36 & 12 & 12 & 30 & 48 & 104 & 20 & 114 & 25 & 25 & 24 & 2 \\
		};
			\end{tikzpicture}
	\caption{Three Generator Pancake Matrix with the first generator being $f_1$ with $n=25$, e.g., the $(19,24)$ entry is 20, which is the order of $f_1f_{19}f_{24}$.}
	\label{f:3matrix}
\end{figure}
\end{center}
\end{ex}

In the next section, we describe the pancake matrix for $B_n$, and make connections to the corresponding pancake graph of $B_n$.
	
		\section{$B_n$ results}\label{s:sbresults}

We now provide a complete description for the order of $f^B_{i}f^B_{j}, 0\leq i,j\leq n-1$ for signed permutations.
	
		\begin{thm}\label{t:mainb}
	If ${m^B_{i-1,j-1}}$ is the order of $f^B_{i-1} f^B_{j-1}$ with $1\leq i< j\leq n$, then
 		\begin{enumerate}
 		\item\label{c:diagonalsb} $m^B_{i-1,i-1}=1$,
 		\item\label{c:symmetryb}  $m^B_{i-1,j-1}=m^B_{j-1,i-1}$,
 		\item\label{c:fourb} If $1< i\leq \lfloor\frac{j}{2}\rfloor$ (with $j\geq4$) then $m^B_{i-1,j-1}=4$.
 		\item\label{c:lemma1b} If $1\leq\lfloor\frac{j}{2}\rfloor<i< j-1$ (with $j\geq4$), then
 		\[m^B_{i-1,j-1}=\begin{cases}
 		2q &\text{ if $r=0$, and}\\
 		2q(q+1)&\text{ if $r\neq0$}
 		 		\end{cases}\] where $d=j-i$, $q=\lfloor\frac{j}{d}\rfloor$, $r=j\pmod{d}$,
 		\item\label{c:offdiagonalsb} If $i=j-1$ (with $j\geq3$) then $m^B_{i-1,j-1}=2j$.
 		\end{enumerate}
 		\end{thm}
 		
 		\begin{proof}
For Case (\ref{c:diagonalsb}), notice that $(f^B_i)^{-1}=f^B_i$, and therefore $m^B_{i,i}=1$.

For Case (\ref{c:symmetryb}), notice that $(f^B_if^B_j)^{-1}=f^B_jf^B_i$, and therefore $m_{i,j}^B=m_{j,i}^B$.

Since elements in $B_{n}$ can be thought of as elements in $B_{n+1}$ that leave $n$ fixed, $M^B_{n}$ can be thought of as the $n\times n$ submatrix of $M^B_{n+1}$ obtained by deleting the last row and the last column of $M^B_{n+1}$. Therefore, it is enough to prove the remaining cases when $j=n$.

For Case (\ref{c:fourb}), notice that the identity permutation $e=[1\;2\;3\;\ldots\;n]$ becomes 
\[
\left[\underline{n}\;\underline{n-1}\;\cdots\;\left(\underline{\left\lfloor\frac{n}{2}\right\rfloor+1}\right)\;\cdots\;1\;2\;\ldots\;j\right]
\] after multiplying it by $f^B_{n-1}f^B_{j-1}$. In other words, $f^B_{n-1}f^B_{j-1}$ will reverse the last $n-j$ symbols of the identity in $B_n$, change their sign and place them at the beginning of the window notation. Since $j\leq\lfloor\frac{n}{2}\rfloor$, the first $j$ characters and the last $n-j$ characters of $e$ (seen as a string in window notation) do not overlap, and thus will behave independently after multiplying by $f^B_{n-1}f^B_{j-1}$. It takes 4 multiplications by $f^B_{n-1}f^B_{j-1}$ for the first $j$ characters of $e$ to return to their original position. Moreover, since $j\leq n-j$, it also take 4 multiplications by $f^B_{n-1}f^B_{j-1}$ for the last $n-j$ characters in $e$ to return to their original position in $e$. Thus $m^B_{n-1,j-1}=4$.

For Case (\ref{c:lemma1b}), let $d=n-j$, $q=\lfloor \frac{n}{d}\rfloor$ and $r=n\pmod{d}$. Notice that any element $w=[w_1\;w_2\;\cdots\;w_n]$ in $B_n$ can be written in the form
\[
\left[\beta_0\alpha_1\beta_1\cdots\alpha_q\beta_q\right],
\] where $\alpha_k$ and $\beta_l$ are substrings of $w$ written in window notation and satisfying $\ell(\alpha_k\beta_k)=d, $ and $\ell(\beta_l)=r$ for $1\leq k\leq q$, $0\leq l\leq q$ (and therefore $\ell(\alpha_k)=d-r$). Here, $\ell(\cdot)$ denotes the length function on strings with brackets ignored. So, for example, $\ell([1\;2\;3])=3$ and $\ell([4\;3\;\underline{2}\;1])=4$.

Multiplying $[\beta_0\alpha_1\beta_1\cdots\alpha_q\beta_q]$ by $f_{n-1}^Bf^B_{j-1}$ gives
\[
\left[\overline{\underline{\beta_q}}\overline{\underline{\alpha_q}}\beta_0\alpha_1\beta_1\cdots\alpha_{q-1}\beta_{q-1}\right],
\] where if $x=x_1x_2\cdots x_m$, $\underline{x}$ and $\overline{x}$ denote $\underline{x_1}\;\underline{x_2}\;\cdots\;\underline{x_m}$ and $x_mx_{m-1}\cdots x_2x_1$, respectively. Since repeated applications of $f_{n-1}^Bf^B_{j-1}$ to $w$ will eventually return $w$ to itself, it follows that the $\alpha$ and $\beta$ segments would return to their original positions.  Therefore the periods of the different $\alpha$ and $\beta$ under multiplication by $f_{n-1}^Bf^B_{j-1}$ have to be the same. The effect of $f_{n-1}^Bf^B_{j-1}$ reverses the last $d$ characters of $w$ and changes their sign. If $r=0$ then there are $q$ $\alpha$ substrings and no $\beta$ substrings. The period of $[\alpha_1\alpha_2\cdots\alpha_q]$ under $f_{n-1}^Bf^B_{j-1}$ is $2q$. Furthermore, if $r\neq0$ then there are $q$ $\alpha$ substrings and $q+1$ $\beta$ substrings, and therefore the period of $[\beta_0\alpha_1\beta_1\ldots\alpha_q\beta_q]$ under $f_{n-1}^Bf^B_{j-1}$ is $2q(q+1)$. This proves Case (\ref{c:lemma1b}).

To prove Case (\ref{c:offdiagonalsb}), notice that $e$ becomes 
\[
\left[\underline{n}\;1\;2\;\ldots\;(n-1)\right]
\] after multiplying it by $f^B_{n-1}f^B_{n-2}$. That is, the effect of multiplying $[w(1)\;w(2)\;\cdots\;w(n)]$ by $f^B_{n-1}f^B_{n-2}$ is to place the last character into the first position and reverse its sign. Thus $2n$ applications of $f^B_{n-1}f^B_{n-2}$ are needed to return the characters of $e$ to its original position. Hence, $m^B_{n-1,n-2}=2n$.
 		\end{proof}

	\begin{ex} In Figure~\ref{f:bmatrix}, we depict $20\times 20$ Coxeter matrix for $B_{20}$.
\begin{center}
\begin{figure}[H]
	\begin{tikzpicture}
	 	\matrix[matrix of math nodes, left delimiter=(, right delimiter=), row sep=1pt, column sep=1pt] (m){ 
			1  & 4  & 4  & 4  & 4  & 4  & 4  & 4  & 4  & 4  & 4  & 4  & 4  & 4  & 4  & 4  & 4  & 4  & 4  & 4 \\ 
			4  & 1  & 6  & 4  & 4  & 4  & 4  & 4  & 4  & 4  & 4  & 4  & 4  & 4  & 4  & 4  & 4  & 4  & 4  & 4 \\ 
			4  & 6  & 1  & 8  & 12 & 4  & 4  & 4  & 4  & 4  & 4  & 4  & 4  & 4  & 4  & 4  & 4  & 4  & 4  & 4 \\ 
			4  & 4  & 8  & 1  & 10 & 6  & 12 & 4  & 4  & 4  & 4  & 4  & 4  & 4  & 4  & 4  & 4  & 4  & 4  & 4 \\ 
			4  & 4  & 12 & 10 & 1  & 12 & 24 & 12 & 12 & 4  & 4  & 4  & 4  & 4  & 4  & 4  & 4  & 4  & 4  & 4 \\ 
			4  & 4  & 4  & 6  & 12 & 1  & 14 & 8  & 6  & 12 & 12 & 4  & 4  & 4  & 4  & 4  & 4  & 4  & 4  & 4 \\ 
			4  & 4  & 4  & 12 & 24 & 14 & 1  & 16 & 40 & 24 & 12 & 12 & 12 & 4  & 4  & 4  & 4  & 4  & 4  & 4 \\ 
			4  & 4  & 4  & 4  & 12 & 8  & 16 & 1  & 18 & 10 & 24 & 6  & 12 & 12 & 12 & 4  & 4  & 4  & 4  & 4 \\ 
			4  & 4  & 4  & 4  & 12 & 6  & 40 & 18 & 1  & 20 & 60 & 8  & 24 & 12 & 12 & 12 & 12 & 4  & 4  & 4 \\ 
			4  & 4  & 4  & 4  & 4  & 12 & 24 & 10 & 20 & 1  & 22 & 12 & 40 & 24 & 6  & 12 & 12 & 12 & 12 & 4 \\ 
			4  & 4  & 4  & 4  & 4  & 12 & 12 & 24 & 60 & 22 & 1  & 24 & 84 & 40 & 24 & 24 & 12 & 12 & 12 & 12\\ 
			4  & 4  & 4  & 4  & 4  & 4  & 12 & 6  & 8  & 12 & 24 & 1  & 26 & 14 & 10 & 8  & 24 & 6  & 12 & 12\\ 
			4  & 4  & 4  & 4  & 4  & 4  & 12 & 12 & 24 & 40 & 84 & 26 & 1  & 28 & 112 & 60 & 40 & 24 & 24 & 12\\ 
			4  & 4  & 4  & 4  & 4  & 4  & 4  & 12 & 12 & 24 & 40 & 14 & 28 & 1  & 30 & 16 & 60 & 40 & 24 & 24\\ 
			4  & 4  & 4  & 4  & 4  & 4  & 4  & 12 & 12 & 6  & 24 & 10 & 112 & 30 & 1  & 32 & 144 & 12 & 40 & 8 \\ 
			4  & 4  & 4  & 4  & 4  & 4  & 4  & 4  & 12 & 12 & 24 & 8  & 60 & 16 & 32 & 1  & 34 & 18 & 84 & 10\\ 
			4  & 4  & 4  & 4  & 4  & 4  & 4  & 4  & 12 & 12 & 12 & 24 & 40 & 60 & 144 & 34 & 1  & 36 & 180 & 84\\ 
			4  & 4  & 4  & 4  & 4  & 4  & 4  & 4  & 4  & 12 & 12 & 6  & 24 & 40 & 12 & 18 & 36 & 1  & 38 & 20\\ 
			4  & 4  & 4  & 4  & 4  & 4  & 4  & 4  & 4  & 12 & 12 & 12 & 24 & 24 & 40 & 84 & 180 & 38 & 1  & 40\\ 
			4  & 4  & 4  & 4  & 4  & 4  & 4  & 4  & 4  & 4  & 12 & 12 & 12 & 24 & 8  & 10 & 84 & 20 & 40 & 1 \\ 
		};
			\end{tikzpicture}
	\caption{Burnt Pancake Matrix with $n=20$. Notice that the matrix is symmetric, the entries in the main diagonal are all 1 and the entries in the off-diagonal are the even integers that are at least 4.}
	\label{f:bmatrix}
\end{figure}
\end{center}
\end{ex}

\subsection{Connection with the burnt pancake graph} The Pancake graph of $S_n$, and in particular its cycle structure, has been extensively studied (see, for example, \cite{Asai2006,HeySud,KF95, KonMed, KKS,KM16, LJD93}). From the results from Theorem~\ref{t:mainb}, one can derive results regarding the cycle structure of the Cayley graph corresponding to $B_n$ generated by $P^B$. Figure~\ref{f:networkb} displays this graph for $B_3$. Indeed, the following theorem, which is a signed version of \cite[Lemma 1]{KM16}, is obtained directly from Theorem~\ref{t:mainb}.

\begin{thm}\label{t:cyclesb}
The Cayley graph of $B_n$ with the generators $P^B$ (burnt pancake graph of $B_n$), with $n\geq2$, contains a maximal set of $\frac{2^nn!}{\ell}$ independent $\ell$-cycles of the form $(f^B_{i}f^B_j)^k$, with $0\leq i< j<n$, $\ell=2k$ and $k=(M^B_{n})_{i+1,j+1}$, the $(i+1,j+1)$ entry in $M^B_n$.
\end{thm}

\begin{proof}
The length of the cycles is given by Theorem~\ref{t:mainb}. Furthermore, every vertex in the burnt pancake graph is incident to exactly one edge corresponding to $f^B_{i}$ and one edge corresponding to $f^B_j$ (with $0\leq i< j< n$), these cycles are independent. Since every signed permutation is the vertex of a cycle of the form $(f^B_{i}f^B_j)^k$, there are $\frac{2^nn!}{\ell}$ of such independent cycles. 
\end{proof}

To illustrate the cycle structure described in the Theorem~\ref{t:cyclesb}, one can look at Figure~\ref{f:networkb} showing the burnt pancake graph of $B_3$. If one considers generators $f^B_1=[\underline{2}\;\underline{1}\;3]$ and $f^B_2=[\underline{3}\;\underline{2}\;\underline{1}]$, then the order of $f^B_1f^B_2$ is 6, and one can indeed notice that there are $\frac{2^3\cdot 3!}{12}=4$ independent cycles labeled with the generators $f^B_1$ (in red), and  $f^B_2$ (in blue). 

It is known that the burnt pancake graph of $B_n$ with $n\geq2$ is an $n$-regular, connected graph that has no triangles nor subgraphs isomorphic to $K_{2,3}$ (see~\cite{K08}). Moreover, if $g(n)$ denotes the diameter of the pancake graph of $B_n$, then $3n/2\leq g(n)\leq 2n-2$ (see~\cite{CohenBlum}). Determining the diameter of the pancake graph of $B_n$ remains an open problem, though exact values are known for $n\leq17$ (see~\cite{Cibulka}).

We recall that a \textit{chord} in a cycle $C$ is an edge not belonging to a $C$ that connects two vertices of $C$. Just in the case for the pancake graph of $S_n$ (see~\cite{KM16}), the cycles described in Theorem~\ref{t:cyclesb} have no chords. We make this formal in the following Lemma.

\begin{lem}\label{lem:nochords}
The cycles described in Theorem~\ref{t:cyclesb} have no chords.
\end{lem}

To prove this lemma, we first recall that the burnt pancake graph of $B_n$ cannot have any simple cycles of length six. 

\begin{lem}[Theorem 10 in \cite{Compeau2011}]\label{lem:nosixcycles}
The \emph{girth} (length of the shortest simple cycle) of the burnt pancake graph of $B_n$ is 8. 
\end{lem}

\noindent\textit{Proof of Lemma~\ref{lem:nochords}.} Let $C=(f_if_j)^{m^{B}_{i,j}}$ be a cycle and suppose that $C$ has a chord. Therefore there exists sign permutations $w_1$ and $w_2$, and $f_k^B\in P_n^B$ such that $w_2f_k^B=w_1$, with $w_1$ and $w_2$ being vertices of $C$. Furthermore, either $w_1(f^B_if^B_j)^s=w_2$ and $(f^B_if^B_j)^sf^B_k=e$, or $w_1(f_i^Bf_j^B)^sf_i^B=w_2$ and $(f^B_if^B_j)^sf_i^Bf^{B}_{k}=e$ with $s<m_{i,j}^B$. Hence, either \[w_2f^B_if^B_jf^B_k=w_1(f^B_if^B_j)^sf^B_if^B_jf^B_k=w_1f^B_if^B_j(f^B_if^B_j)^sf^B_k=w_1f^B_if_j^B,\text{ or}\]

 \[w_2f^B_jf^B_if^B_k=w_1(f_i^Bf_j^B)^sf_i^Bf^B_jf^B_if^B_k=w_1f_i^Bf_j^B(f_i^Bf^B_j)^{s}f^B_if^B_k=w_1f^B_if_j^B.\]
 
 Therefore, there exist a 6-cycle of the form $f^{B}_{i}f^{B}_{j}f^{B}_{k}f^{B}_{j}f^{B}_{i}f^{B}_{k}$ or of the form $(f^{B}_{i}f^{B}_{j}f^{B}_{k})^{2}$. This contradicts Lemma~\ref{lem:nosixcycles}, and therefore no such cycle $C$ exists.
\hfill$\square$

\subsection{Reflections} We now describe the set of \textit{burnt pancake reflections} $$T_B^{\pm}=\{wf^B_iw^{-1}\mid 0\leq i\leq n-1,w\in B_n\}.$$ We recall that any element in the set of reflections $T_B=\{ws^B_iw^{-1}\mid w\in B_n,0\leq i\leq n-1\}$ for signed permutations has the following form (see~\cite[Proposition 8.1.5]{BjornerBrenti}):
\[
\{(i,j)(\underline{i},\underline{j})\mid 1\leq i<|j|\leq n\}\cup\{(i,\underline{i})\mid 1\leq i\leq n\}.
\]

If $t\in T_B^{\pm}$ then $t=wf_i^Bw^{-1}$ for some $w\in B_n$ and $0\leq i\leq n-1$. If $w=[w_1\;w_2\;\cdots\;w_n]$, then from $wf^{B}_iw^{-1}=t$ we have
\[
    [\underline{w_{i+1}}\;\underline{w_{i}}\;\cdots\;\underline{w_{1}}\;w_{i+2}\;w_{i+3}\;\cdots\;w_n]=tw,
\] and so $t=(w_1,\underline{w_{i+1}})(w_2,\underline{w_{i}})\cdots(w_{i+1},\underline{w_1})$. In terms of notation, if a $w_j<0$, $1\leq j\leq n$, then $\underline{w_{j}}=-w_j>0$. Therefore, 

\begin{equation}\label{eq:breflections}
    T_B^{\pm}=\{(w_1,\underline{w_{i+1}})(w_2,\underline{w_{i}})\cdots(w_{i+1},\underline{w_1})\mid  w_i\in[\pm n], 0\leq i< n, w_{a}\neq w_{b}\text{ if }a\neq b\}
\end{equation}

In terms of comparing $T_B$ and $T^{\pm}_B$, we notice that any permutation of the form $(i,\underline{i})$ is in both sets. However, permutations of the form $(i,j)(i,\underline{j})$ with $1\leq i<|j|\leq n$ are not. 

As for the number of burnt reflections, from the description in (\ref{eq:breflections}), one gets

\begin{cor}
$\left|T^{\pm}_B\right|=\displaystyle\sum_{i=1}^n\binom{n}{i}2^{\lfloor i/2\rfloor}$.
\end{cor}

\section{Acknowledgments}\label{s:ack}

The authors are grateful to Ivars Peterson for an introduction to the subject whose talk in 2014 at the EPaDel sectional meeting inspired this paper. We also thank Jacob Mooney and Kyle Yohler for independently writing computer code to verify our results.

\end{document}